\documentclass[onecolumn, a4size, 11pt]{IEEEtran}
\usepackage{amsmath}
\usepackage{amssymb}
\usepackage{amsfonts}
\usepackage{graphicx}
\usepackage{epsfig}
\usepackage{subfigure}
\usepackage{psfrag}

\title{User Cooperation in Wireless Powered Communication Networks\footnote{The authors are with the Department of Electrical and Computer Engineering, National University of Singapore (e-mail: elejhs@nus.edu.sg, elezhang@nus.edu.sg).}}

\author{Hyungsik Ju and Rui Zhang}

\begin{document}

\maketitle \thispagestyle{empty}

\setlength{\baselineskip}{1.3\baselineskip}
\newtheorem{definition}{\underline{Definition}}[section]
\newtheorem{fact}{Fact}
\newtheorem{assumption}{Assumption}
\newtheorem{theorem}{\underline{Theorem}}[section]
\newtheorem{lemma}{\underline{Lemma}}[section]
\newtheorem{corollary}{\underline{Corollary}}[section]
\newtheorem{proposition}{\underline{Proposition}}[section]
\newtheorem{example}{\underline{Example}}[section]
\newtheorem{remark}{\underline{Remark}}[section]
\newtheorem{algorithm}{\underline{Algorithm}}[section]
\newcommand{\mv}[1]{\mbox{\boldmath{$ #1 $}}}

\begin{abstract}
This paper studies user cooperation in the emerging wireless powered communication network (WPCN) for throughput optimization. For the purpose of exposition, we consider a two-user WPCN, in which one hybrid access point (H-AP) broadcasts wireless energy to two distributed users in the downlink (DL) and the users transmit their independent information using their individually harvested energy to the H-AP in the uplink (UL) through time-division-multiple-access (TDMA). We propose user cooperation in the WPCN where the user which is nearer to the H-AP and has a better channel for DL energy harvesting as well as UL information transmission uses part of its allocated UL time and DL harvested energy to help to relay the far user's information to the H-AP, in order to achieve more balanced throughput. We maximize the weighted sum-rate (WSR) of the two users by jointly optimizing the time and power allocations in the network for both wireless energy transfer in the DL and wireless information transmission and relaying in the UL. Simulation results show that the proposed user cooperation scheme can effectively improve the achievable throughput in the WPCN with desired user fairness.
\end{abstract}

\section{Introduction}\label{Sec:Introduction}

Energy harvesting has recently received a great deal of attention in wireless communication since it provides virtually perpetual energy supplies to wireless networks through scavenging energy from the environment. In particular, harvesting energy from the far-field radio-frequency (RF) signal transmissions is a promising solution, which opens a new avenue for the unified study of wireless energy transfer (WET) and wireless information transmission (WIT) as radio signals are able to carry energy and information at the same time.

There are two main paradigms of research along this direction. One line of work aims to characterize the fundamental trade-offs in simultaneous WET and WIT with the same transmitted signal in the so-called simultaneous wireless information and power transfer (SWIPT) systems (see e.g., \cite{Zhou}-\cite{Zhang} and the references therein). Another line of research focuses on designing a new type of wireless network termed wireless powered communication network (WPCN) in which wireless terminals communicate using the energy harvested from WET (see e.g., \cite{Huang}-\cite{Ju_Throughput_for_WPCN}).

In our previous work \cite{Ju_Throughput_for_WPCN}, we studied a typical WPCN model, in which one hybrid access-point (H-AP) coordinates WET/WIT to/from a set of distributed users in the downlink (DL) and uplink (UL) transmissions, respectively. It has been shown in \cite{Ju_Throughput_for_WPCN} that the WPCN suffers from a so-called ``doubly near-far'' problem, which occurs when a far user from the H-AP receives less wireless energy than a near user in the DL, but needs to transmit with more power in the UL for achieving the same communication performance due to the doubled distance-dependent signal attenuation over both the DL and UL. As a result, unfair rate allocations among the near and far users are inured when their sum-throughput is maximized. In \cite{Ju_Throughput_for_WPCN}, we proposed to assign shorter/longer time to the near/far users in their UL WIT to solve the doubly near-far problem, which is shown to achieve more fair rate allocations among the users in a WPCN.

On the other hand, user cooperation is an effective way to improve the capacity, coverage, and/or diversity performance in conventional wireless communication systems. Assuming constant energy supplies at user terminals, cooperative communication has been thoroughly investigated in the literature under various protocols such as decode-and-forward and amplify-and-forward (see e.g., \cite{Sendonaris}, \cite{Liang} and the references therein). Recently, cooperative communication has been studied in energy harvesting wireless communication and SWIPT systems (see e.g. \cite{C_Huang}-\cite{Nasir}). However, how to exploit user cooperation in the WPCN to overcome the doubly near-far problem and further improve the network throughput and user fairness still remains unknown, which motivates this work.

\begin{figure}[!t]
   \centering
   \includegraphics[width=0.5\columnwidth]{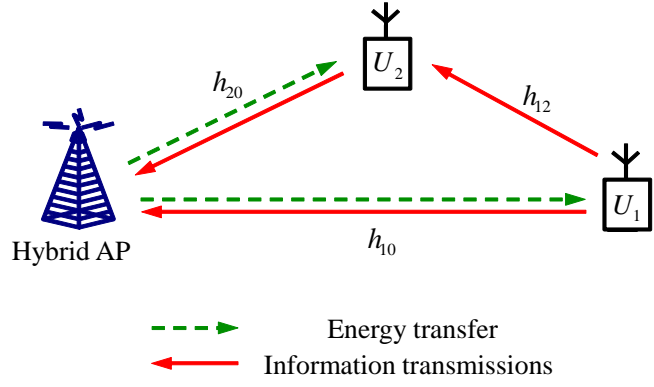}
   \caption{A two-user wireless powered communication network (WPCN) with DL WET and UL WIT via user cooperation.}
   \label{Fig_SystemModel}
\end{figure}

In this paper, we study user cooperation in the WPCN for throughput optimization. For the purpose of exposition, we consider a two-user WPCN, as shown in Fig. \ref{Fig_SystemModel}, where one H-AP broadcasts wireless energy to two distributed users with different distances in the DL, and the two users transmit their independent information using individually harvested energy to the H-AP in the UL through time-division-multiple-access (TDMA). To enable user cooperation, we propose that the near user which has a better channel than the far user for both DL WET and UL WIT uses part of its allocated UL time and DL harvested energy to first help to relay the information of the far user to the H-AP and then uses the remaining time and energy to transmit its own information. Under this protocol, we characterize the maximum weighted sum-rate (WSR) of the two users by jointly optimizing the time and power allocations in the network for both WET in the DL and WIT in the UL, subject to a given total time constraint. The achievable throughput gain in the WPCN by the proposed user cooperation scheme is shown both analytically and through simulations over the baseline scheme in \cite{Ju_Throughput_for_WPCN} without user cooperation.

The rest of this paper is organized as follows. Section \ref{SystemModel} presents the system model of the WPCN with user cooperation. Section \ref{ResourceAlloc} presents the time and power allocation problem to maximize the WSR in the WPCN, and compares the solutions and achievable throughput regions with versus without user cooperation. Section \ref{SimulationResult} presents more simulation results under practical fading channel setups. Finally, Section \ref{Conclusion} concludes the paper.

\section{System Model}\label{SystemModel}
As shown in Fig. \ref{Fig_SystemModel}, this paper considers a two-user WPCN with WET in the DL and WIT in the UL. The network consists of one hybrid access point (H-AP) and two users (e.g., sensors) denoted by $U_1$ and $U_2$, respectively, operating over the same frequency band. The H-AP and the users are assumed to be each equipped with one antenna. Furthermore, it is assumed that the H-AP has a constant energy supply (e.g., battery), whereas $U_1$ and $U_2$ need to replenish energy from the received signals broadcast by the H-AP in the DL, which is then stored and used to maintain their operations (e.g., sensing and data processing) and also communicate with the H-AP in the UL.

We assume without loss of generality that $U_2$ is nearer to the H-AP than $U_1$, and hence denote the distance between the H-AP and $U_1$, that between the H-AP and $U_2$, and that between the $U_1$ and $U_2$ as $D_{10}$, $D_{20}$, and $D_{12}$, respectively, with $D_{10} \ge D_{20}$. We also assume that $D_{12} \le D_{10}$ so that $U_2$ can more conveniently decode the information sent by $U_1$ than the H-AP, to motivate the proposed user cooperation to be introduced next. Assuming that the channel reciprocity holds between the DL and UL, then the DL channel from the H-AP to user $U_i$ and the corresponding reversed UL channel are both denoted by a complex random variable $\tilde{h}_{i0}$ with channel power gain $h_{i0} = {| \tilde{h}_{i0} |^2}$, $i = 1,2$, which in general should take into account the distance-dependent signal attenuation and long-term shadowing as well as the short-term fading. In addition, the channel from $U_1$ to $U_2$ is denoted by a complex random variable $\tilde{h}_{12}$ with channel power gain $h_{12} = {| \tilde{h}_{12} |^2}$. If only the distance-dependent signal attenuation is considered, we should have $h_{10} \le h_{12}$ and $h_{10} \le h_{20}$ due to the assumptions of $D_{10} \ge D_{20}$ and $D_{10} \ge D_{12}$. Furthermore, we consider block-based transmissions over quasi-static flat-fading channels, where $h_{10}$, $h_{20}$, and $h_{12}$ are assumed to remain constant during each block transmission time, denoted by $T$, but can vary from one block to another. In each block, it is further assumed that the H-AP has the perfect knowledge of $h_{10}$, $h_{20}$, and $h_{12}$, and $U_2$ knows perfectly $h_{12}$.

\begin{figure}[!t]
   \centering
   \includegraphics[width=0.5\columnwidth]{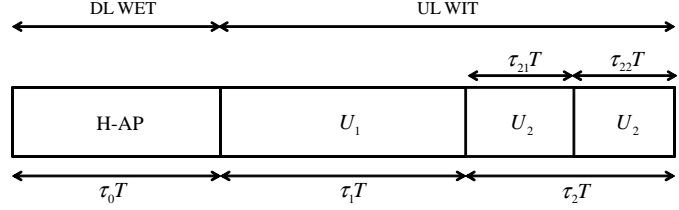}
   \caption{Transmission protocol for WPCN with user cooperation.}
   \label{Fig_FrameStructure}
\end{figure}

We extend the harvest-then-transmit protocol proposed in \cite{Ju_Throughput_for_WPCN} for the two-user WPCN to enable user cooperation, as shown in Fig. \ref{Fig_FrameStructure}. In each block, during the first $\tau_0 T$ amount of time, $0 < \tau_0 <1$, the H-AP broadcasts wireless energy to both $U_1$ and $U_2$ in the DL with fixed transmit power $P_0$. The far user $U_1$ then transmits its information with average power $P_1$ during the subsequent $\tau_1 T$ amount of time in the UL, $0 < \tau_1 <1$, using its harvested energy, and both the H-AP and $U_2$ decode the received signal from $U_1$. To overcome the doubly near-far problem \cite{Ju_Throughput_for_WPCN}, during the remaining $(1-\tau_0-\tau_1)T$ amount of time in each block, the near user $U_2$ first relays the far user $U_1$'s information and then transmits its own information to the H-AP using its harvested energy with average power $P_{21}$ over $\tau_{21} T$ amount of time and with average power $P_{22}$ over $\tau_{22} T$ amount of time, respectively, where $\tau_{21} + \tau_{22} = \tau_2$. Note that we have a total time constraint given by
\begin{equation}\label{Eq_SumTime}
   {\sum\limits_{i = 0}^2 {{\tau _i}}  = \tau_0 + \tau_1 + \tau_{21} + \tau_{22} \le 1.}
\end{equation}
For convenience, we noramlize $T = 1$ in the sequel without loss of generality.

During the DL phase, the transmitted complex baseband signal of the H-AP in one block of interest is denoted by an arbitrary random signal, $x_0$, satisfying $\mathbb E [ {{{\left| {{x_0}} \right|}^2}} ] = P_0$. The received signal at $U_i$, $i = 1,2$, is then expressed as
\begin{equation}\label{Eq_ReceivedSignal_Sensors}
   {{y_i^{ (0)  }} = \sqrt{h_{i\,0}} {x_0} + {z_i}, \,\,\,\,\,\, i = 1, 2,}
\end{equation}
where $y_r^{(k)}$ denotes the received signal at $U_r$ during $\tau_k$, with $k \in \left\{ {0,1,21,22} \right\}$ and $r \in \left\{ {0,1,2} \right\}$ (with $U_0$ denoting the H-AP in the sequel). In (\ref{Eq_ReceivedSignal_Sensors}), ${z_i}$ denotes the received noise at $U_i$ which is assumed to be ${z_i} \sim {\mathcal{CN}}\left( {0,\sigma_i^2} \right)$, $i = 1,2$, where ${\mathcal{CN}}( {\mu ,\sigma^2})$ stands for a circularly symmetric complex Gaussian (CSCG) random variable with mean $\mu$ and variance $\sigma^2$. It is assumed that $P_0$ is sufficiently large such that the energy harvested due to the receiver noise is negligible and thus is ignored. Hence, the amount of energy harvested by each user in the DL can be expressed as (assuming unit block time, i.e., $T = 1$)
\begin{equation}\label{Eq_HarvestedEnergy}
   {{E_i} = {\zeta_i}{P_0}{h_{i0}}{\tau _{0}}, \,\,\,\,\,\, i = 1, 2,}
\end{equation}
where $0 < \zeta_i < 1$, $i = 1, 2$, is the energy conversion efficiency at the receiver of $U_i$.

After the DL phase, each user uses a fixed portion of its harvested energy, denoted by $\eta_i$, with $0 < \eta_i \le 1$, $i = 1, 2$, for the UL transmissions, i.e., transmitting own information (by both $U_1$ and $U_2$) or relaying the other user's information (by $U_2$ only) to the H-AP. Within the first $\tau_1$ amount of time allocated to $U_1$, the average transmit power of $U_1$ is given by
\begin{equation}\label{Eq_TxPower_U1}
   {{P_1} = \frac{{{\eta _1}{E_1}}}{{{\tau _1}}} = {\eta _1}{\zeta _1}{P_0}{h_{10}}\frac{{{\tau _0}}}{{{\tau _1}}}.}
\end{equation}
We denote $x_1$ as the complex baseband signal transmitted by $U_1$ with power $P_1$, which is assumed to be Gaussian, i.e., ${x_1} \sim {\mathcal{CN}}\left( {0,P_1} \right)$. The received signals at the H-AP and $U_2$ in this UL slot for $U_1$ are expressed, respectively, as
\begin{equation}\label{Eq_ReceivedSignal_Slot1}
   {y_i^{(1)} = \sqrt {{h_{1i}}} \, {x_1} + {z_i}, \,\, i=0,\,2,}
\end{equation}
where ${z_0} \sim {\mathcal{CN}}\left( {0,\sigma_0^2} \right)$ denotes the receiver noise at the H-AP.

During the last $\tau_2$ amount of time allocated to $U_2$, the total energy consumed by $U_2$ for transmitting its own  information and relaying the decoded information for $U_1$ should be no larger than $\eta_2 E_2$, i.e.,
\begin{equation}\label{Eq_TxPower_U2}
   {\tau_{21} P_{21} + \tau_{22} P_{22} \le \eta_2 {\zeta_2}{P_0}{h_{20}}{\tau _{0}}.}
\end{equation}
We denote the complex basedband signals transmitted by $U_2$ for relaying $U_1$'s information and transmitting its own information as $x_{21}$ with power $P_{21}$ and $x_{22}$ with power $P_{22}$, respectively, where ${x_{21}} \sim {\mathcal{CN}}\left( {0,P_{21}} \right)$ and ${x_{22}} \sim {\mathcal{CN}}\left( {0,P_{22}} \right)$. During $\tau_{21}$ and $\tau_{22}$ amount of time allocated to $U_2$, the corresponding received signals at the H-AP can be expressed as
\begin{equation}\label{Eq_ReceivedSignal_Slot2}
   {y_0^{(2i)} = \sqrt {{h_{20}}} \, {x_{2i}} + {z_0}, \,\, i = 1, \,2.}
\end{equation}

Denote the time allocations to DL WET and UL WIT as $\boldsymbol{\tau} = [\tau_0, \,\, \tau_1, \,\, \tau_{21}, \,\, \tau_{22}]$, and the transmit power values of $U_1$ and $U_2$ for UL WIT as ${\bf{P}} = [P_1, \,\, P_{21}, \,\, P_{22}]$. From \cite{Liang}, the achievable rate of $U_1$ for a given pair of $\boldsymbol{\tau}$ and ${\bf{P}}$ can be expressed from (\ref{Eq_ReceivedSignal_Slot1}) and (\ref{Eq_ReceivedSignal_Slot2}) as
\begin{equation}\label{Eq_AchievableRate_U1}
   {R_1\left( { \boldsymbol{\tau}, {\bf{P}} } \right) = \min \left[ {{R_1^{(10)}}\left( { \boldsymbol{\tau}, {\bf{P}} } \right) + {R_1^{(20)}}\left( { \boldsymbol{\tau}, {\bf{P}} } \right),\,\,{R_1^{(12)}}\left( { \boldsymbol{\tau}, {\bf{P}} } \right)} \right],}
\end{equation}
with ${R_1^{(10)}}\left( { \boldsymbol{\tau}, {\bf{P}} } \right)$, ${R_1^{(20)}}\left( { \boldsymbol{\tau}, {\bf{P}} } \right)$, and ${R_1^{(12)}}\left( { \boldsymbol{\tau}, {\bf{P}} } \right)$ denoting the achievable rates of the transmissions from $U_1$ to the H-AP, from $U_2$ to the H-AP, and from $U_1$ to $U_2$, respectively, which are given by
\begin{equation}\label{Eq_R1_10}
   {{R_1^{(10)}}\left( { \boldsymbol{\tau}, {\bf{P}} } \right) = {\tau _1}{\log _2}\left( {1 + \frac{{{P_1}{h_{10}}}}{{{\sigma_0 ^2}}}} \right) ,}
\end{equation}
\begin{equation}\label{Eq_R1_12}
   {{R_1^{(12)}}\left( { \boldsymbol{\tau}, {\bf{P}} } \right) = {\tau _1}{\log _2}\left( {1 + \frac{{{P_1}{h_{12}}}}{{{\sigma_2 ^2}}}} \right) ,}
\end{equation}
\begin{equation}\label{Eq_R1_20}
   {{R_1^{(20)}}\left( { \boldsymbol{\tau}, {\bf{P}} } \right) = {\tau _{21}}\,{\log _2}\left( {1 + \frac{{{P_{21}}{h_{20}}}}{{{\sigma_0 ^2}}}} \right).}
\end{equation}
Furthermore, the achievable rate of $U_2$ is expressed from (\ref{Eq_ReceivedSignal_Slot2}) as
\begin{equation}\label{Eq_AchievableRate_U2}
   {{R_{2}}\left( { \boldsymbol{\tau}, {\bf{P}} } \right) = {\tau _{22}}\,{\log _2}\left( {1 + \frac{{{P_{22}}{h_{20}}}}{{{\sigma_0 ^2}}}} \right).}
\end{equation}

\section{Optimal Time and Power Allocations in WPCN with User Cooperation}\label{ResourceAlloc}
In this section, we study the joint optimization of the time allocated to the H-AP, $U_1$, and $U_2$, i.e., $\boldsymbol{\tau}$, and power allocations of the users, i.e., $\bf{P}$, to maximize the weighted sum-rate (WSR) of the two users in UL transmission.

Let $\boldsymbol{\omega} = [\omega_1, \,\,\, \omega_2]$ with $\omega_1$ and $\omega_2$ denoting the given non-negative rate weights for $U_1$ and $U_2$, respectively. The WSR maximization problem is then formulated from (\ref{Eq_AchievableRate_U1})-(\ref{Eq_AchievableRate_U2}) as
\[
   {\rm{(P1)}}:\,\,\,\,\mathop {\max }\limits_{{\boldsymbol{\tau}}, \,\, {\bf{P}}} \,\,\,{\omega _1}{R_1}\left( {{\boldsymbol{\tau}} ,{\bf{P}}} \right) + {\omega _2}{R_2}\left( {{\boldsymbol{\tau}} ,{\bf{P}}} \right)
\]
\[
   \,\,\, {\rm{s}}{\rm{.t}}{\rm{.}}\,\,\,\,{\rm{(1)}}, \,\,\, (4), \,\,\,{\rm{and}}\,\,\,{\rm{(6),}}
\]
\[
   \,\,\,\,\,\,\,\,\,\,\,\,\,\,\,\,\,\,\,\,\,\,\,\,\,\,\,\,\,\,\,\,\,\,\,\,\,\,\,\,\,\,\,\,\,\,\,\,\,\,\,\,\,\,\,\,\, {\tau _0} \ge 0, \,\,\, {\tau _1} \ge 0, \,\,\, {\tau _{21}} \ge 0, \,\,\, {\tau _{22}} \ge 0,
\]
\[
   \,\,\,\,\,\,\,\,\,\,\,\,\,\,\,\,\,\,\,\,\,\,\,\,\,\,\,\,\,\,\,\,\,\,\,\,\,\,\, {P_1} \ge 0, \,\,\, {P_{21}} \ge 0, \,\,\, {P_{22}} \ge 0.
\]
Notice that if we set $\tau_{21} = 0$ and $P_{21} = 0$, then (P1) reduces to the special case of WPCN without user cooperation studied in \cite{Ju_Throughput_for_WPCN}, i.e., the near user $U_2$ only transmits its own information to the H-AP, but does not help the far user $U_1$ for relaying its information to the H-AP.

Note that (P1) can be shown to be non-convex in the above form. To make this problem convex, we change the variables as $t_{21} = \frac{\tau_{21} P_{21}}{\eta_2 \zeta_2 h_{20} P_0}$ and $t_{22} = \frac{\tau_{22} P_{22}}{\eta_2 \zeta_2 h_{20} P_0}$. Since ${P_1} = {\eta _1}{\zeta _1}{P_0}{h_{10}}\frac{{{\tau _0}}}{{{\tau _1}}}$ as given in (\ref{Eq_TxPower_U1}), ${R_1^{(10)}}\left( { \boldsymbol{\tau}, {\bf{P}} } \right)$, ${R_1^{(12)}}\left( { \boldsymbol{\tau}, {\bf{P}} } \right)$, ${R_1^{(20)}}\left( { \boldsymbol{\tau}, {\bf{P}} } \right)$, and ${R_{2}}\left( { \boldsymbol{\tau}, {\bf{P}} } \right)$ in (\ref{Eq_R1_10})-(\ref{Eq_AchievableRate_U2}) can be re-expressed as functions of ${\bf{t}} = [{\boldsymbol{\tau}}, \,\, t_{21}, \,\, t_{22}]$, i.e.,
\begin{equation}\label{Eq_R1_10_new}
   {{R_1^{(10)}}\left( { {\bf{t}} } \right) = {\tau _1}{\log _2}\left( {1 + \rho_1^{(10)} \frac{{{\tau _0}}}{{{\tau _1}}}} \right),}
\end{equation}
\begin{equation}\label{Eq_R1_12_new}
   {{R_1^{(12)}}\left( { {\bf{t}} } \right) = {\tau _1}{\log _2}\left( {1 + \rho_1^{(12)} \frac{{{\tau _0}}}{{{\tau _1}}}} \right),}
\end{equation}
\begin{equation}\label{Eq_R1_20_New}
   {{R_1^{(20)}}\left( { {\bf{t}} } \right) = {\tau _{21}}\,{\log _2}\left( {1 + \rho_2\frac{t_{21}}{\tau_{21}}} \right),}
\end{equation}
\begin{equation}\label{Eq_AchievableRate_U2_New}
   {{R_{2}}\left( { {\bf{t}} } \right) = {\tau _{22}}\,{\log _2}\left( {1 + \rho_2 \frac{t_{22}}{\tau_{22}}} \right),}
\end{equation}
where $\rho_1^{(10)} = h_{10}^2 \frac{\eta_1 \zeta_1 P_0}{\sigma_0^2}$, $\rho_1^{(12)} = {h_{10}}{h_{12}} \frac{\eta_1 \zeta_1 P_0}{\sigma_2^2}$, and $\rho_2 = h_{20}^2  \frac{\eta_2 \zeta_2 P_0}{\sigma_0^2}$. Furthermore, we introduce a new variable $\bar R$ defined as ${\bar R} = \min\,\,[{R_1^{(10)}}\left( { {\bf{t}} } \right) + {R_1^{(20)}}\left( { {\bf{t}} } \right), \,\, {R_1^{(12)}}\left( { {\bf{t}} } \right)]$. It then follows that $\bar R \le {R_1^{(10)}}\left( { {\bf{t}} } \right) + {R_1^{(20)}}\left( { {\bf{t}} } \right)$ and $\bar R \le {R_1^{(12)}}\left( { {\bf{t}} } \right)$. Accordingly, (P1) can be equivalently reformulated as
\[\]
\begin{equation}\label{Eq_Problem_Epigraph}
   {\rm{(P2)}}:\,\,\,\,\mathop {\max }\limits_{\bar R,\,\, {\bf{t}}} \,\,\,\,\,{\omega _1}\bar R + {\omega _2}{R_2}\left( {{\bf{t}}} \right) \,\,\,\,\,\,\,\,\,\,\,\,\,\,\,\,\,\,\,\,\,\,\,\,\,
\end{equation}
\begin{equation}\label{Eq_SumTimeConstraint_New}
   \,\,\,\,\,\,\,\,\,\,\,\,\,\,\,\,\,\,\,\,\,\,\,\,\,\,\,\,\,\,\,\,\,\,\,\,\,\,\,\,\,\,\,\,\,\,\, {\rm{s}}{\rm{.t}}{\rm{.}} \,\,\,\,\,\,\, \tau_0 + \tau_1 + \tau_{21} + \tau_{22} \le 1,
\end{equation}
\begin{equation}\label{Eq_SumEnergyConstraint}
   \,\,\,\,\,\,\,\,\,\,\,\,\,\,\,\,\,\,\,\,\,\,\,\,\,\,\,\,\,\,\,\,\,\,\,\, {t_{21} + t_{22} \le \tau_0,} \,\,\,\,\,\,\,\,\,\,\,\,\,\,\,\,\,\,\,\,\,\,\,\,
\end{equation}
\begin{equation}\label{Eq_U1_Rate_Constraint}
   \,\,\,\,\,\,\,\,\,\,\,\,\,\,\,\,\,\,\,\,\,\,\,\,\,\,\,\,\,\,\,\,\,\,\,\,\,\,\,\,\,\,\,\,\,\,\,\,\,\,\,\,\,\,\,\,\,\,\,\,\,\,\,\, \bar R \le {R_1^{(10)}}\left( { {\bf{t}} } \right) + {R_1^{(20)}}\left( { {\bf{t}} } \right),
\end{equation}
\begin{equation}\label{Eq_U2_Rate_Constraint}
   \,\,\,\,\,\,\,\,\,\,\,\,\,\,\,\,\,\,\,\,\,\,\,\,\,\,\,\,\,\, \bar R \le {R_1^{(12)}}\left( { {\bf{t}} } \right), \,\,\,\,\,\,\,\,\,\,\,\,\,\,\,
\end{equation}
where the time constraint in (\ref{Eq_SumEnergyConstraint}) can be shown to be equivalent to the power constraint originally given in (\ref{Eq_TxPower_U2}). It is worth noting that $t_{21}$ and $t_{22}$ denote the amount of time in the DL slot duration $\tau_0$ in which the harvested energy by $U_2$ is later allocated to relay $U_1$'s information and transmit its own information in the UL, respectively. By introducing the new variables $t_{21}$ and $t_{22}$ in $\bf{t}$ and $\bar{R}$, joint time and power allocation in problem (P1) is converted to time allocation only in problem (P2).

Note that ${R_1^{(10)}}\left( { {\bf{t}} } \right)$, ${R_1^{(20)}}\left( { {\bf{t}} } \right)$, ${R_1^{(12)}}\left( { {\bf{t}} } \right)$, and ${R_{2}}\left( { {\bf{t}} } \right)$ are all monotonically increasing functions over each element of $(\tau_0,\tau_1)$, $(\tau_0,\tau_1)$, $(t_{21},\tau_{21})$, and $(t_{22},\tau_{22})$, respectively. Let the optimal solution of (P2) be denoted by ${\bf{t}}^* = [{\boldsymbol{\tau}}^*, t_{21}^*,t_{22}^*]$ $=$ $[\tau_0^*, \tau_1^*, \tau_{21}^*, \tau_{22}^*, t_{21}^*,t_{22}^*]$. Then, it can be easily verified that $t_{21}^* + t_{22}^* = \tau_0^*$ must hold (otherwise, we can always increase ${R_{2}}\left( { {\bf{t}} } \right)$ by increasing $t_{22}$ to improve the weighted sum-rate). Similarly, it can also be verified that
\begin{equation}\label{Eq_Prop_Opt_Inequality}
   {{R_1^{(10)}}\left( { {\bf{t}}^* } \right) + {R_1^{(20)}}\left( { {\bf{t}}^* } \right) \le {R_1^{(12)}}\left( {  {\bf{t}}^* } \right),}
\end{equation}
since, otherwise, we can allocate part of $\tau_{21}$ (or $t_{21}$) to $\tau_{22}$ (or $t_{22}$) until the equality holds, which will result in an increased ${R_2}\left( { {\boldsymbol{\tau}}, {\bf{t}} } \right)$ without reducing ${R_1}\left( { {\boldsymbol{\tau}}, {\bf{t}} } \right)$.

\begin{lemma}\label{Lemma_Concavity}
   ${R_1^{(10)}}\left( { {\bf{t}} } \right)$, ${R_1^{(12)}}\left( { {\bf{t}} } \right)$, ${R_1^{(20)}}\left( { {\bf{t}} } \right)$, and ${R_{2}}\left( { {\bf{t}} } \right)$ are all concave functions of $\bf{t}$.
\end{lemma}
\begin{proof}
   Please refer to Appendix \ref{App_Proof_Lemma_Concavity}.
\end{proof}

From Lemma \ref{Lemma_Concavity}, it follows that the objective function of (P2) is a concave function of $\bf{t}$, and so are the functions at the right-hand side of both (\ref{Eq_U1_Rate_Constraint}) and (\ref{Eq_U2_Rate_Constraint}). Furthermore, the constraints in (\ref{Eq_SumTimeConstraint_New}) and (\ref{Eq_SumEnergyConstraint}) are both affine. Therefore, problem (P2) is a convex optimization problem, and furthermore it can be verified that (P2) satisfies the Slater's condition \cite{ConvexOptimization}; hence, it can be solved by the Lagrange duality method, shown as follows. From (\ref{Eq_Problem_Epigraph})-(\ref{Eq_U2_Rate_Constraint}), the Lagrangian of (P2) is given by
\[
    \mathcal L\left( {\bar R, {\bf{t}}, \boldsymbol{\lambda} } \right) = {\omega _1}\bar R + {\omega _2}{R_2}\left( {{\bf{t}}} \right)  - \lambda_1 \left( \tau_0 + \tau_1 + \tau_{21} + \tau_{22} - 1 \right) - \lambda_2 \left( {t_{21} + t_{22} - \tau_0} \right)\,\,\,\,\,\,\,\,\,\,\,\,\,\,\,\,\,\,
\]
\begin{equation}\label{Eq_Lagrangian}
     - \lambda_3 \left( \bar R - {R_1^{(10)}}\left( {  {\bf{t}} } \right) - {R_1^{(20)}}\left( { {\bf{t}} } \right) \right) - \lambda_4 \left( \bar R - {R_1^{(12)}}\left( { {\bf{t}} } \right) \right),
\end{equation}
where ${\boldsymbol{\lambda}} = [\lambda_1, \,\, \lambda_2, \,\, \lambda_3, \,\, \lambda_4]$ denote the Lagrange multipliers associated with the constraints in (\ref{Eq_SumTimeConstraint_New}), (\ref{Eq_SumEnergyConstraint}), (\ref{Eq_U1_Rate_Constraint}), and (\ref{Eq_U2_Rate_Constraint}), respectively. Notice that $\lambda_3 + \lambda_4 \ge \omega_1$ must hold; otherwise, the Lagrnagian will go unbounded from above with $\bar{R}\rightarrow \infty$. The dual function of problem (P2) is then given by
\begin{equation}\label{Eq_DualFunction}
   \mathcal G \left( \boldsymbol{\lambda} \right) = \mathop {\max }\limits_{ {{\bf{t}}} \in \mathcal D, \, \bar{R} \ge 0} \,\,\, \mathcal L\left( \bar R, {{\bf{t}},{\boldsymbol{\lambda}}} \right),
\end{equation}
where $\mathcal D$ is the feasible set of ${\bf{t}}$ specified by $t_{21} \ge 0$, $t_{22} \ge 0$, and ${\boldsymbol{\tau}} \ge 0$ (`$\ge$' here denotes the component-wise inequality). The dual problem of (P2) is thus given by $\mathop {\min }\limits_{{\boldsymbol{\lambda}}  \ge 0, \lambda_3+\lambda_4  \ge \omega_1 } \,\,\, \mathcal G\left( {\boldsymbol{\lambda}} \right)$. The optimal solution ${\bf{t}}^*$ can be obtained if the optimal dual solution $\boldsymbol{\lambda}^*$ is found by solving the dual problem of (P2).

\begin{proposition}\label{Proposition_Opt_Solution}
   Given positive weights $\omega_1 > 0$ and $\omega_2 > 0$, the optimal solution to (P2), ${\bf{t}}^* = [{\boldsymbol{\tau}^*}, t_{21}^*, t_{22}^*]$, is given by
   \begin{equation}\label{Eq_Prop_Opt_tau}
      {{\boldsymbol{\tau}^*} = \left[ \frac{\left( {\sqrt{b^2 - 4ac}} - b \right)\tau_1^*}{2a} , \,\,\, \frac{\tau_0^*}{z_1^*}, \,\,\, \frac{\rho_2 t_{21}^*}{z_{21}^*}, \,\,\, \frac{\rho_2 t_{22}^*}{z_{22}^*} \right],}
   \end{equation}
   \begin{equation}\label{Eq_Prop_Opt_E}
      {\left[t_{21}^*,t_{22}^*\right] = \left[ \left( \frac{\lambda_3^*\tau_{21}^*}{\lambda_2^* \ln 2} - \frac{\tau_{21}^*}{\rho_2} \right)^{+} , \left( \frac{\omega_2^*\tau_{22}^*}{\lambda_2^* \ln 2} - \frac{\tau_{22}^*}{\rho_2} \right)^{+}  \right],}
   \end{equation}
   with $(x)^{+} \buildrel \Delta \over =  \max(0,x)$, and $\lambda_1^* > 0$, $\lambda_2^* > 0$, $\lambda_3^* \ge 0$, and $\lambda_4^* \ge 0$ denoting the optimal dual solutions. Moreover, $a$, $b$, and $c$ in (\ref{Eq_Prop_Opt_tau}) are given, respectively, by
   \begin{equation}\label{Eq_Prop_a}
      {a = \left( \lambda_1^* - \lambda_2^* \right) \rho_1^{(10)} \rho_1^{(12)},}
   \end{equation}
   \begin{equation}\label{Eq_Prop_b}
      {b = \left( \lambda_1^* - \lambda_2^* \right) \left( \rho_1^{(10)} + \rho_1^{(12)}\right) - \omega_1 \rho_1^{(10)} \rho_1^{(12)},}
   \end{equation}
   \begin{equation}\label{Eq_Prop_c}
      {c = \lambda_1^* - \lambda_2^* - \lambda_3^* \rho_1^{(10)} - \lambda_4 \rho_1^{(12)}.}
   \end{equation}
   Finally, $z_1^*$, $z_{21}^*$, and $z_{22}^*$ in (\ref{Eq_Prop_Opt_tau}) are solutions of $\lambda_3^* f( \rho_1^{(10)} z ) + \lambda_4^* f( \rho_1^{(12)} z ) = \lambda_1^* \ln 2$, $f\left( z \right) = \frac{\lambda_1^* \ln 2}{\lambda_3^*}$, and $f\left( z \right) = \frac{\lambda_1^* \ln 2}{\omega_2}$, respectively, where
   \begin{equation}\label{Eq_Prop_Function_z}
      {f\left( z \right) \buildrel \Delta \over =   \ln \left( {1 + z} \right) - \frac{z}{{1 + z}}.}
   \end{equation}
\end{proposition}
\begin{proof}
   Please refer to Appendix \ref{App_Proof_Prop_Opt_Solution}.
\end{proof}

\begin{table}[!t]
\renewcommand{\arraystretch}{1.3}
\caption{Algorithm to solve $(\rm{P1}).$}
\label{Table_P1} \centering
   \begin{tabular}{|p{3.4in}|}
   \hline
      1)  \textbf{Initialize} $\boldsymbol{\lambda} \ge 0$

      2) \textbf{Repeat}
         \begin{itemize}
               \item[1.] Initialize $k=0$, ${{\boldsymbol{\tau}}} = {{\boldsymbol{\tau}}}^{(0)}$, ${{\bf{t}}} = {{\boldsymbol{\tau}}}^{(0)}$.

               \item[2.] \textbf{Repeat}
                  \begin{itemize}
                     \item[(1)] Obtain $[\tau_0^{(k+1)}, \,\, t_{21}^{(k+1)}, \,\, t_{22}^{(k+1)}]$ from (\ref{Eq_Prop_Opt_tau}) and (\ref{Eq_Prop_Opt_E}) \

                         with given $[\tau_1^{(k)}, \,\, \tau_{21}^{(k)}, \,\, \tau_{22}^{(k)}]$.

                     \item[(2)] Obtain $[\tau_1^{(k+1)}, \,\, \tau_{21}^{(k+1)}, \,\, \tau_{22}^{(k+1)}]$ from (\ref{Eq_Prop_Opt_tau}) and (\ref{Eq_Prop_Opt_E}) \

                         with given $[\tau_0^{(k+1)}, \,\, t_{21}^{(k+1)}, \,\, t_{22}^{(k+1)}]$.
                  \end{itemize}

               \item[3.] \textbf{until} ${{\bf{t}}}^{\star}$ converges to a predetermined accuracy.

               \item[4.] Compute ${\bar R}^{\star} = \min\,\,[{R_1^{(10)}}\left( { {\bf{t}}^{\star} } \right) + {R_1^{(20)}}\left( { {\bf{t}}^{\star} } \right), \,\, {R_1^{(12)}}\left( { {\bf{t}}^{\star} } \right)]$.

               \item[5.] Update $\boldsymbol{\lambda}$ subject to $\lambda_3+\lambda_4 \ge \omega_1$ using the ellipsoid method and the subgradient of $\mathcal{G} \left( {\boldsymbol{\lambda}} \right)$ given by (\ref{Eq_Subgradient1})-(\ref{Eq_Subgradient4}).
         \end{itemize}

      3) \textbf{Until} Stopping criteria of the ellipsoid method is met.

      4) \textbf{Set} $P_{21}^* = \eta_2 \zeta_2 h_{20} P_0 \frac{t_{21}^*}{\tau_{21}^*}$  and $P_{22}^* = \eta_2 \zeta_2 h_{20} P_0 \frac{t_{22}^*}{\tau_{22}^*}$.
            \\
   \hline
   \end{tabular}
\end{table}

According to Proposition \ref{Proposition_Opt_Solution}, we can obtain ${\bf{t}}^*$ as follows. Denote ${\bf{t}}^{\star}$ and $\bar R ^{\star}$ as the maximizer of $\mathcal L\left( {\bar R, {\bf{t}}, \boldsymbol{\lambda} } \right)$ in (\ref{Eq_Lagrangian}) for a given $\boldsymbol{\lambda}$. We can first obtain ${\bf{t}}^{\star}$ by iteratively optimizing between $[\tau_0, \,\, t_{21}, \,\, t_{22}]$ and $[\tau_1, \,\, \tau_{21}, \,\, \tau_{22}]$ using (\ref{Eq_Prop_Opt_tau}) and (\ref{Eq_Prop_Opt_E}) with one of them being fixed at one time, until they both converge. Then we compute $\bar R^{\star} = \min [{R_1^{(10)}}\left( { {\bf{t}}^{\star} } \right) + {R_1^{(20)}}\left( {{\bf{t}}^{\star} } \right), {R_1^{(12)}}\left( { {\bf{t}}^{\star} } \right)]$. With $\mathcal G (\boldsymbol{\lambda})$ obtained for each given $\boldsymbol{\lambda}$, the optimal $\boldsymbol{\lambda}^*$ minimizing $\mathcal G (\boldsymbol{\lambda})$ can then be found by updating $\boldsymbol{\lambda}$ using sub-gradient based algorithms, e.g., the ellipsoid method \cite{LectureNote}, with the sub-gradient of $\mathcal G (\boldsymbol{\lambda})$, denoted as $\boldsymbol{\nu} = [\nu_1, \,\, \nu_2, \,\, \nu_3, \,\, \nu_4]$, given by
\[\]
\begin{equation}\label{Eq_Subgradient1}
   \nu_1 = \tau_0^{\star} + \tau_1^{\star} + \tau_{21}^{\star} + \tau_{22}^{\star} - 1,
\end{equation}
\begin{equation}\label{Eq_Subgradient2}
   \nu_2 = t_{21}^{\star} + t_{22}^{\star} - \tau_0^{\star},
\end{equation}
\begin{equation}\label{Eq_Subgradient3}
   \nu_3 = \bar R^{\star} - {R_1^{(10)}}\left( { {\bf{t}}^{\star} } \right) + {R_1^{(20)}}\left( { {\bf{t}}^{\star} } \right),
\end{equation}
\begin{equation}\label{Eq_Subgradient4}
   \nu_4 = \bar R^{\star} - {R_1^{(12)}}\left( { {\bf{t}}^{\star} } \right).
\end{equation}
Once ${\boldsymbol{\lambda}}^{*}$ and the corresponding ${\bf{t}}^* = {\bf{t}}^{\star}$ are obtained, the optimal power allocation solution at $U_2$ for (P1) is obtained as $P_{21}^* = \eta_2 \zeta_2 h_{20} P_0 \frac{t_{21}^*}{\tau_{21}^*}$  and $P_{22}^* = \eta_2 \zeta_2 h_{20} P_0 \frac{t_{22}^*}{\tau_{22}^*}$. To summarize, one algorithm to solve problem (P1) is given in Table \ref{Table_P1}.

Fig. \ref{Fig_ThroughputRegion_k_05_Various_a} shows the achievable throughput regions of the two-user WPCN with user cooperation by solving (P1) with different user rate weights as compared to that by the baseline scheme in \cite{Ju_Throughput_for_WPCN} without user cooperation, for different values of path-loss exponent, $\alpha$. It is assumed that $D_{10} = 10$m, and $D_{12} = D_{20} = 5$m. The channel power gains in the network are modeled as $h_{ij} = 10^{-3} \theta_{ij} D_{ij}^{-\alpha}$, $ij \in \left\{ {10,20,12} \right\}$, for distance $D_{ij}$ in meter, with the same path-loss exponent $\alpha$ and $30$dB signal power attenuation for both users at a reference distance of $1$m, where $\theta_{ij}$ represents the additional channel short-term fading. We ignore the effects of short-term fading in this case by setting $\theta_{10} = \theta_{20} = \theta_{12} =1$, to focus on the effect of the doubly near-far problem due to distance-dependent attenuation only. Moreover, it is assumed that $P_0 = 30$dBm and the bandwidth is $1$MHz. The AWGN at the receivers of the H-AP and $U_2$ is assumed to have a white power spectral density of $-160$dBm/Hz. For each user, it is assumed that $\eta_1 = \eta_2 = 0.5$ and $\zeta_1 = \zeta_2 = 0.5$.

\begin{figure}[!t]
   \centering
   \includegraphics[width=0.7\columnwidth]{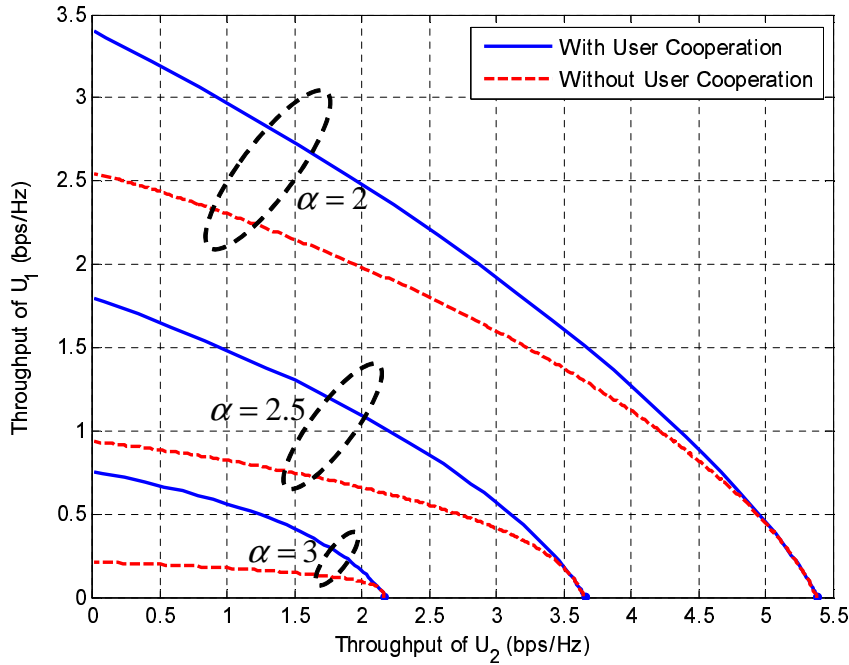}
   \caption{Throughput region comparison for WPCN with versus without user cooperation.}
   \label{Fig_ThroughputRegion_k_05_Various_a}
\end{figure}

From Fig. \ref{Fig_ThroughputRegion_k_05_Various_a}, it is observed that the throughput region of WPCN with user cooperation is always larger than that without user cooperation, which is expected as the latter case only corresponds to a suboptimal solution of (P1) in general. Let $\delta = R_{1,\max}^{(wc)}/R_{1,\max}^{(nc)}$, with $R_{1,\max}^{(wc)}$ and $R_{1,\max}^{(nc)}$ denoting the maximum achievable throughput of the far user $U_1$ in the WPCN with and without user cooperation, respectively. It is then inferred from Fig. \ref{Fig_ThroughputRegion_k_05_Various_a} that $\delta = 1.33$, $1.92$, and $3.60$ when $\alpha = 2$, $2.5$, and $3$, respectively, which implies that user cooperation in the WPCN is more beneficial in improving the far user's rate as $\alpha$ increases, i.e., when the doubly near-far problem is more severe. This is because the achievable rate for the direct link from $U_1$ to the H-AP decreases more significantly than that of the other two links over $\alpha$.

\begin{figure}[!t]
   \centering
   \includegraphics[width=0.7\columnwidth]{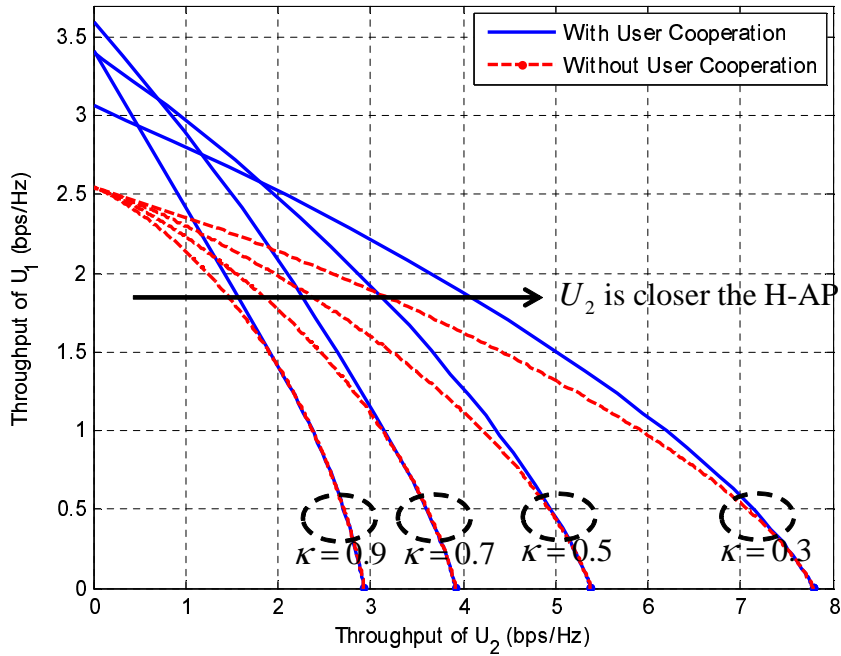}
   \caption{Throughput region comparison for WPCN with versus without user cooperation with $\alpha = 2$.}
   \label{Fig_ThroughputRegion_a_2_Various_k}
\end{figure}

Next, Fig. \ref{Fig_ThroughputRegion_a_2_Various_k} compares the achievable throughput regions of WPCN with versus without user cooperation with $\alpha = 2$. In this case, the H-AP and the two users are assumed to lie on a straight  line with $D_{20} = \kappa D_{10}$ and $D_{12} = (1-\kappa)D_{10}$, $0 < \kappa < 1$. It is observed that when $\kappa$ is not large (i.e., $\kappa \le 0.7$), $R_{1,\max}^{(wc)}$ decreases with decreasing $\kappa$. This is because when the near user $U_2$ moves more away from the far user $U_1$ (and thus closer to the H-AP), the degradation of $R_1^{(12)}({\bf{t}}^*)$ for the $U_1$-to-$U_2$ link with decreasing $\kappa$ is more significant than the improvement in $R_1^{(20)}({\bf{t}}^*)$ of the $U_2$-to-H-AP link since $R_1^{(10)}({\bf{t}}^*) + R_1^{(20)}({\bf{t}}^*) \le R_1^{(12)}({\bf{t}}^*)$ with the optimal time allocations ${\bf{t}}^*$. On the other hand, when $\kappa$ is larger than a certain threshold (e.g., $\kappa=0.9$), $R_{1,\max}^{(wc)}$ decreases with increasing $\kappa$ since in this case not only the far user $U_1$, but also the relatively nearer user $U_2$ suffers from the significant signal attenuation from/to the H-AP.

\begin{figure}[!t]
   \centering
   \includegraphics[width=0.7\columnwidth]{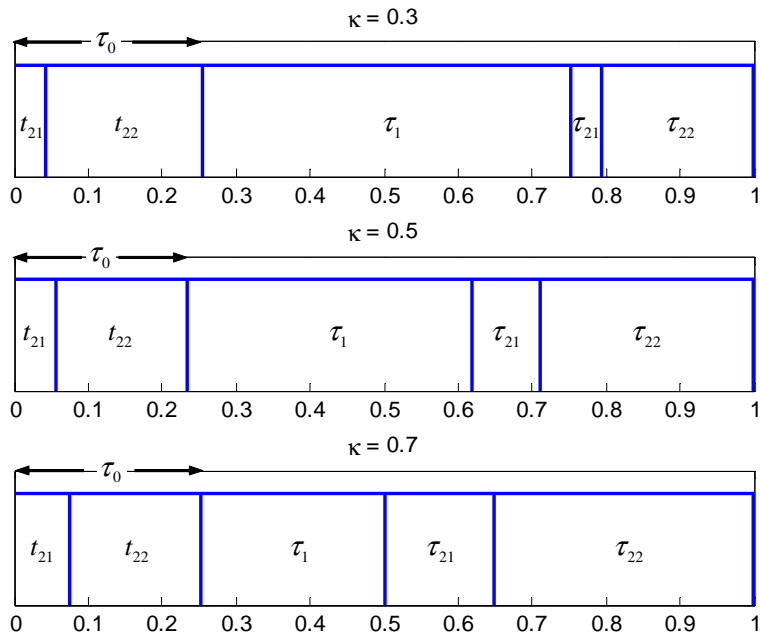}
   \caption{Optimal time allocations in ${\bf{t}}^*$ for different values of $\kappa$ when $\alpha = 2$ and $R_1({\bf{t}}^*) = R_2({\bf{t}}^*)$.}
   \label{Fig_CommonThroughput_TimeAlloc}
\end{figure}

Finally, Fig. \ref{Fig_CommonThroughput_TimeAlloc} shows the optimal time allocations in ${\bf{t}}^*$ for (P2) when $R_1({\bf{t}}^*) = R_2({\bf{t}}^*)$, i.e., the common-throughput \cite{Ju_Throughput_for_WPCN} is maximized,\footnote{The common-throughput can be obtained by searching over $\boldsymbol{\omega}$, for which one algorithm is provided in \cite{Ju_Throughput_for_WPCN}.} with $\alpha = 2$ and $\kappa = 0.3$, $0.5$, $0.7$. It is observed that $\tau_1^*$ decreases but both $\tau_{21}^*$ and $\tau_{22}^*$ increase with increasing $\kappa$. This is because when the near user $U_2$ moves more away from the H-AP, $U_2$ suffers from more severe signal attenuation as $\kappa$ increases, and thus it is necessary to allocate more time to $U_2$ for both transmitting own information and relaying information for $U_1$ in order to maximize the common throughput with $R_1({\bf{t}}^*) = R_2({\bf{t}}^*)$.

\section{Simulation Result}\label{SimulationResult}
In this section, we compare the maximum common throughput in the WPCN with versus without user cooperation under the practical fading channel setup, while the other system parameters are set similarly as for Figs. \ref{Fig_ThroughputRegion_k_05_Various_a} and \ref{Fig_ThroughputRegion_a_2_Various_k}. The short-term fading in the network is assumed to be Rayleigh distributed, and thus $\theta_{10}$, $\theta_{20}$, and $\theta_{12}$ in the previously given channel models are exponentially distributed with unit mean.

\begin{figure}
   \centering
   \includegraphics[width=0.7\columnwidth]{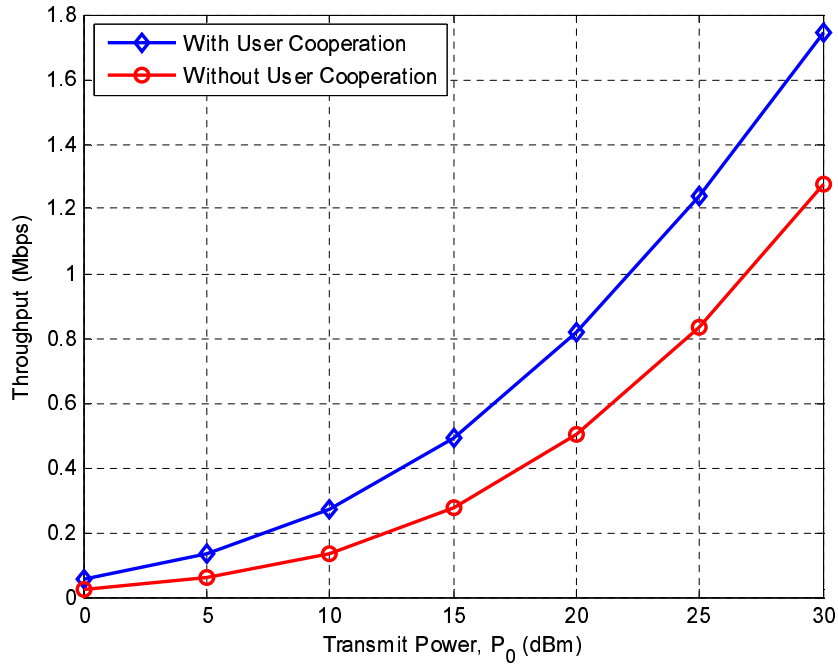}
   \caption{Maximum common-throughput versus $P_0$ with $\alpha =2$ and $\kappa = 0.5$. }
   \label{Fig_AvgThroughput_vs_Pa}
\end{figure}

Fig. \ref{Fig_AvgThroughput_vs_Pa} shows the maximum average common-throughput versus the transmit power of H-AP, i.e., $P_0$ in dBm, with $\alpha = 2$ and $\kappa = 0.5$. It is observed that the maximum common-throughput in the WPCN with user cooperation is notably larger than that without user cooperation, especially when $P_0$ becomes large. This result shows the effectiveness of the proposed user cooperation in the WPCN to further improve both the throughput and user fairness as compared to the baseline scheme in \cite{Ju_Throughput_for_WPCN} with optimized time allocation only but without user cooperation.

Fig. \ref{Fig_AvgThroughput_vs_k} shows the maximum average common-throughput versus different values of $\kappa$ with $P_0 = 30$dBm. It is observed that the maximum common-throughput in the WPCN with user cooperation is always larger than that without user cooperation. Furthermore, the common-throughput in the WPCN with user cooperation first increases over $\kappa$, but decreases with increasing $\kappa$ when $\kappa$ is larger than a certain threshold. The threshold value of $\kappa$ that maximizes the average common-throughput of the WPCN with user cooperation is observed to increase over $\alpha$.

\begin{figure}
   \centering
   \includegraphics[width=0.7\columnwidth]{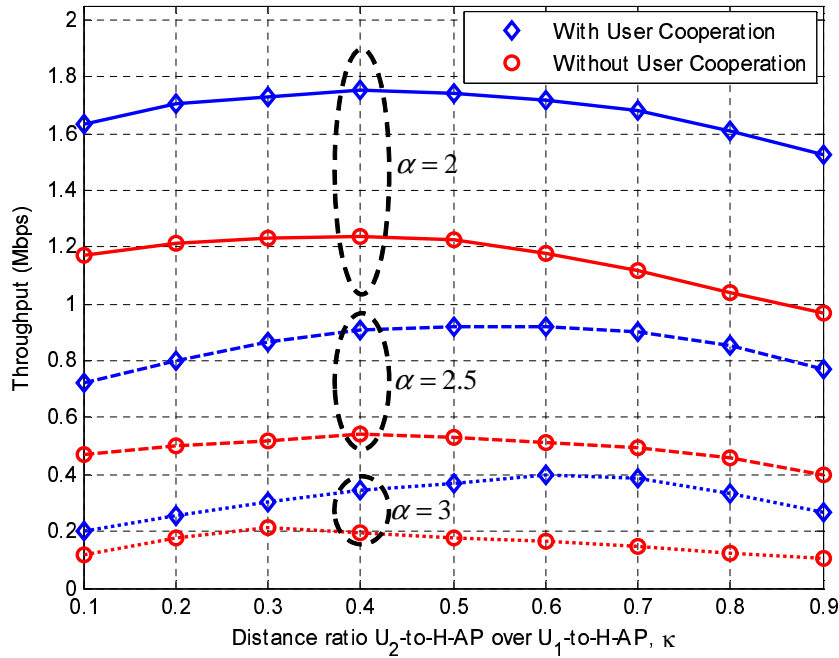}
   \caption{Maximum common-throughput versus $\kappa$ with $P_0=30$dBm and $\alpha =2$, $2.5$, $3$. }
   \label{Fig_AvgThroughput_vs_k}
\end{figure}

\section{Conclusion} \label{Conclusion}
This paper studied a two-user WPCN in which user cooperation is jointly exploited with resources (time, power) allocation to maximize the network throughput and at the same time achieve desired user fairness by overcoming the doubly near-far problem. We characterized the maximum WSR in the WPCN with user cooperation via a problem reformulation and applying the tools from convex optimization. By comparing the achievable throughput regions as well as the maximum common-throughput in the WPCN with versus without user cooperation, it is shown by extensive simulations that the proposed user cooperation is effective to improve both the throughput and user fairness. In future work, we will extend the results of this paper to other setups, e.g., when there are more than two users, alternative relaying schemes are applied, and/or other performance metrics are considered.

\appendices

   \section{Proof of Lemma \ref{Lemma_Concavity}}\label{App_Proof_Lemma_Concavity}
   To prove Lemma \ref{Lemma_Concavity}, we use the following lemma.

   \begin{lemma}\label{Lemma_Concavity_General}
      For two variables $x \ge 0$ and $x_2 \ge 0$, a function $g(x,y)$ defined as
      \begin{equation}\label{Eq_App_Lemma_Concavity}
         {g\left( x_1,x_2 \right) \buildrel \Delta \over =   \left\{ {\begin{array}{*{20}{c}}
         {x_1\log \left( {1 + \alpha \frac{x_2}{x_1}} \right)}  \\
         0  \\
         \end{array}\,\,\,\begin{array}{*{20}{c}}
         {,\,\,\,x > 0}  \\
         {,\,\,\,x = 0}  \\
         \end{array}} \right.}
      \end{equation}
      is a jointly concave function of both $x_1$ and $x_2$.
   \end{lemma}
   \begin{proof}
      Please refer to Appendix \ref{App_Proof_Lemma_Concavity_General}.
   \end{proof}

   Note that ${R_1^{(10)}}\left( { {\bf{t}} } \right)$, ${R_1^{(12)}}\left( { {\bf{t}} } \right)$, ${R_1^{(20)}}\left( { {\bf{t}} } \right)$, and ${R_{2}}\left( { {\bf{t}} } \right)$ are all functions of only two elements in ${\bf{t}} = [ \tau_0,$ $\tau_1$, $\tau_{21},$ $\tau_{22},$ $t_{21},$ $t_{22} ]$, all of which have the equivalent form as (\ref{Eq_App_Lemma_Concavity}). Therefore, ${R_1^{(10)}}\left( { {\bf{t}} } \right)$, ${R_1^{(12)}}\left( { {\bf{t}} } \right)$, ${R_1^{(20)}}\left( { {\bf{t}} } \right)$, and ${R_{2}}\left( { {\bf{t}} } \right)$ are all concave functions of ${\bf{t}}$. This completes the proof.

   \section{Proof of Lemma \ref{Lemma_Concavity_General}}\label{App_Proof_Lemma_Concavity_General}
   Denote the Hessian of $g(x_1,x_2)$ defined in (\ref{Eq_App_Lemma_Concavity}) as 
   \begin{equation}\label{Eq_App_fungtion_g}
      {{\nabla ^2}g\left( {{x_1},{x_2}} \right) = \left[ {{d_{i,j}}} \right],\,\,\,i,j \in \left\{ {1,2} \right\},}
   \end{equation}
   where $d_{i,j}$ is given by 
   \begin{equation}\label{Eq_App_Hessian_function_g}
      {{d_{i,j}} = \left\{ {\begin{array}{*{20}{c}}
      { - \frac{{{\alpha ^2}x_2^2}}{{x_1^3{{\left( {1 + \alpha \frac{{{x_2}}}{{{x_1}}}} \right)}^2}}}\,\,\,,\,\,\,i = j = 1\,\,\,}  \\
      {\frac{{{\alpha ^2}{x_2}}}{{x_1^2{{\left( {1 + \alpha \frac{{{x_2}}}{{{x_1}}}} \right)}^2}}}\,\,\,,\,\,\,\,\,\,\,\,\,i \ne j\,\,\,}  \\
      { - \frac{{{\alpha ^2}}}{{{x_1}{{\left( {1 + \alpha \frac{{{x_2}}}{{{x_1}}}} \right)}^2}}}\,\,\,,\,\,\,\,\,i = j = 2.}  \\
      \end{array}\,\,} \right.}
   \end{equation}
   
   Given an arbitrary real vector ${\bf{v}} = [v_1, \,\, v_2]^T$, it can be shown from (\ref{Eq_App_fungtion_g}) and (\ref{Eq_App_Hessian_function_g}) that 
   \begin{equation}\label{Eq_App_Concavity_function_g}
      {{{\bf{v}}^T}{\nabla ^2}g\left( {{x_1},{x_2}} \right){\bf{v}} =  - \frac{{{\alpha ^2}}}{{{x_1}{{\left( {1 + \alpha \frac{{{x_2}}}{{{x_1}}}} \right)}^2}}}{\left( {\frac{{{x_2}}}{{{x_1}}}{v_1} - {v_2}} \right)^2} \le 0,}
   \end{equation}
   i.e., ${\nabla ^2}g\left( {{x_1},{x_2}} \right)$ is a negative semi-definite matrix. Therefore, $g\left( {{x_1},{x_2}} \right)$ is a jointly concave function of both $x_1$ and $x_2$ \cite{ConvexOptimization}. This completes Lemma \ref{Lemma_Concavity_General}.

   \section{Proof of Proposition \ref{Proposition_Opt_Solution}}\label{App_Proof_Prop_Opt_Solution}   
   Since (P2) is a convex optimization problem for which the strong duality holds, the Karush-Kuhn-Tucker (KKT) conditions are both necessary and sufficient for the global optimality of (P2), which is shown below.
   \[\]
   \begin{equation}\label{Eq_App_Prop_KKT_Deri_Rbar}
      {\frac{\partial }{{\partial \bar R^*}} \mathcal L = {\omega _1} - {\lambda _3^*} - {\lambda _4^*} = 0,}
   \end{equation}
   \begin{equation}\label{Eq_App_Prop_KKT_Deri_tau0}
      {\frac{\partial }{{\partial \tau_0^*}} \mathcal L = \frac{1}{{\ln 2}}\left( {\frac{{\lambda _3^*\rho _1^{(10)}}}{{1 + \rho _1^{(10)}\frac{{\tau _0^*}}{{\tau _1^*}}}} + \frac{{\lambda _4^*\rho _1^{(12)}}}{{1 + \rho _1^{(12)}\frac{{\tau _0^*}}{{\tau _1^*}}}}} \right) - \lambda _1^* + \lambda _2^* = 0,}
   \end{equation}
   \begin{equation}\label{Eq_App_Prop_KKT_Deri_tau1}
      {\frac{\partial }{{\partial \tau_1^*}} \mathcal L = \frac{{\lambda _3^*}}{{\ln 2}}\left( {\ln \left( {1 + \rho _1^{(10)}\frac{{\tau _0^*}}{{\tau _1^*}}} \right) + \frac{{\rho _1^{(10)}\frac{{\tau _0^*}}{{\tau _1^*}}}}{{1 + \rho _1^{(10)}\frac{{\tau _0^*}}{{\tau _1^*}}}}} \right) + \frac{{\lambda _4^*}}{{\ln 2}}\left( {\ln \left( {1 + \rho _1^{(12)}\frac{{\tau _0^*}}{{\tau _1^*}}} \right) - \frac{{\rho _1^{(12)}\frac{{\tau _0^*}}{{\tau _1^*}}}}{{1 + \rho _1^{(12)}\frac{{\tau _0^*}}{{\tau _1^*}}}}} \right) - \lambda _1^* = 0,}
   \end{equation}
   \begin{equation}\label{Eq_App_Prop_KKT_Deri_tau21}
      {\frac{\partial }{{\partial \tau_{21}^*}} \mathcal L = \frac{{\lambda _3^*}}{{\ln 2}}\left( {\ln \left( {1 + {\rho _2}\frac{{t_{21}^*}}{{\tau _{21}^*}}} \right) - \frac{{{\rho _2}\frac{{t_{21}^*}}{{\tau _{21}^*}}}}{{1 + {\rho _2}\frac{{t_{21}^*}}{{\tau _{21}^*}}}}} \right) - \lambda _1^* = 0,}
   \end{equation}
   \begin{equation}\label{Eq_App_Prop_KKT_Deri_tau22}
      {\frac{\partial }{{\partial \tau_{22}^*}} \mathcal L = \frac{{{\omega _2}}}{{\ln 2}}\left( {\ln \left( {1 + {\rho _2}\frac{{t_{22}^*}}{{\tau _{22}^*}}} \right) - \frac{{{\rho _2}\frac{{t_{22}^*}}{{\tau _{22}^*}}}}{{1 + {\rho _2}\frac{{t_{22}^*}}{{\tau _{22}^*}}}}} \right) - \lambda _1^* = 0,}
   \end{equation}
   \begin{equation}\label{Eq_App_Prop_KKT_Deri_E21}
      {\frac{\partial }{{\partial t_{21}^*}} \mathcal L = \frac{{\lambda _3^*}}{{\ln 2}}\frac{{{\rho _2}}}{{1 + {\rho _2}\frac{{t_{21}^*}}{{\tau _{21}^*}}}} - \lambda _2^* = 0,}
   \end{equation}
   \begin{equation}\label{Eq_App_Prop_KKT_Deri_E22}
      {\frac{\partial }{{\partial t_{22}^*}} \mathcal L = \frac{{{\omega _2}}}{{\ln 2}}\frac{{{\rho _2}}}{{1 + {\rho _2}\frac{{t_{22}^*}}{{\tau _{22}^*}}}} - \lambda _2^* = 0,}
   \end{equation}
   \begin{equation}\label{Eq_App_Prop_KKT_Slack1}
      \lambda_1^* \left(\tau_0^* + \tau_1^* + \tau_{21}^* + \tau_{22}^* - 1\right) = 0,
   \end{equation}
   \begin{equation}\label{Eq_App_Prop_KKT_Slack2}
      \lambda_2^* \left( t_{21}^* + t_{22}^* - \tau_0^*\right) = 0,
   \end{equation}
   \begin{equation}\label{Eq_App_Prop_KKT_Slack3}
      \lambda_3^* \left( \bar R^* - {R_1^{(10)}}\left( { {\bf{t}}^* } \right)  - {R_1^{(20)}}\left( { {\bf{t}}^* } \right) \right) = 0,
   \end{equation}
   \begin{equation}\label{Eq_App_Prop_KKT_Slack4}
      \lambda_4^* \left( \bar R^* - {R_1^{(12)}}\left( { {\bf{t}}^* } \right) \right) = 0.
   \end{equation}
   
   Since $t_{21}^* + t_{22}^* = \tau_0^*$ must hold for (P2), we assume without loss of generality that $\lambda_2^* > 0$ ($\lambda_2^* = 0$ only when $t_{21}^* = t_{22}^* = 0$ in (\ref{Eq_App_Prop_KKT_Deri_E21}) and (\ref{Eq_App_Prop_KKT_Deri_E22}), i.e., no harvested energy at $U_2$ is used for UL WIT). Furthermore, it can be easily verified that $\tau_0^* + \tau_1^* + \tau_{21}^* + \tau_{22}^* = 1$ must hold for (P2) and thus we can also assume that $\lambda_1^* > 0$ with no loss of generality ($\lambda_1^* = 0$ only when $\tau_0^* = t_{21}^* = t_{22}^* = 0$ from (\ref{Eq_App_Prop_KKT_Deri_tau21}) and (\ref{Eq_App_Prop_KKT_Deri_tau22}) i.e., no energy is transferred by the H-AP). 
   
   Changing variable as $z_1 = \frac{\tau_0^*}{\tau_1^*}$ and after mathematically manipulations, (\ref{Eq_App_Prop_KKT_Deri_tau0}) can be modified as $a z_1^2 + b z_1 + c = 0$, where $a$, $b$, and $c$ are given in (\ref{Eq_Prop_a})-(\ref{Eq_Prop_c}). Since $\tau_0^* \ge 0$, we thus have $\tau_0^* = \frac{\tau_1^*}{2a} \left( {\sqrt{b^2 - 4ac}} - b \right)$ in (\ref{Eq_Prop_Opt_tau}) from quadratic formula. Furthermore, with $z_1 = \frac{\tau_0^*}{\tau_1^*}$ and from (\ref{Eq_App_Prop_KKT_Deri_tau1}), we also have $\lambda_3^* f( \rho_1^{(10)} z_1 ) + \lambda_4^* f( \rho_1^{(12)} z_1 ) = \lambda_1^* \ln 2$, where $f(z)$ is given in (\ref{Eq_Prop_Function_z}). It is worth noting that $f(z)$ given in (\ref{Eq_Prop_Function_z}) is a monotonically increasing function of $z \ge 0$ where $f(0) = 0$, and so is $\lambda_3^* f( \rho_1^{(10)} z_1) + \lambda_4^* f( \rho_1^{(12)} z_1 )$. Therefore, there exists a unique $z_1^*$ satisfying $\lambda_3^* f( \rho_1^{(10)} z_1 ) + \lambda_4^* f( \rho_1^{(12)} z_1 ) = \lambda_1^* \ln 2$ for given $\lambda_1 > 0$, from which we have $\tau_1^*$ given in (\ref{Eq_Prop_Opt_tau}). Similarly, by changing variables as $z_{21}^* = \rho_2\frac{t_{21}^*}{\tau_{21}^*}$ and $z_{22}^* = \rho_2 \frac{t_{22}^*}{\tau_{22}^*}$ in (\ref{Eq_App_Prop_KKT_Deri_tau21}) and (\ref{Eq_App_Prop_KKT_Deri_tau22}), respectively, we can obtain unique $z_{21}^*$ and $z_{22}^*$ which are solutions of $f\left( z_{21} \right) = \frac{\lambda_1^* \ln 2}{\lambda_3^*}$ and $f\left( z_{22} \right) = \frac{\lambda_1^* \ln 2}{\omega_2}$, from which $\tau_{21}^*$ and $\tau_{22}^*$ can be obtained. Finally, we have $t_{21}^*$ and $t_{22}^*$ in (\ref{Eq_Prop_Opt_E}) from (\ref{Eq_App_Prop_KKT_Deri_E21}) and (\ref{Eq_App_Prop_KKT_Deri_E22}), respectively. This completes the proof of Proposition \ref{Proposition_Opt_Solution}.


\clearpage
   

\begin{thebibliography}{1}
\bibliographystyle{IEEEbib}


   \bibitem{Zhou}
      X. Zhou, R. Zhang, and C. K. Ho, ``Wireless information and power transfer: architecture design and rate-energy tradeoff,'' \emph{IEEE Trans. Commun.}, vol. 61, no. 11, pp.4757-4767, Nov. 2013.

   \bibitem{Liu}
      L. Liu, R. Zhang, and K. C. Chua, ``Wireless information transfer with opportunistic energy harvesting,'' \emph{IEEE Trans. Wireless Commun.}, vol. 12, no. 1, pp. 288-300, Jan. 2013.

   \bibitem{Zhang}
      R. Zhang and C. K. Ho, ``MIMO broadcasting for simultaneous wireless information and power transfer,'' \emph{IEEE Trans. Wireless Commun.}, vol. 12, no. 5, pp. 1989-2001, May 2013.

   \bibitem{Huang}
      K. Huang and V. K. N. Lau, ``Enabling wireless power transfer in cellular networks: architecture, modeling and deployment,'' \emph{IEEE Trans. Wireless Commun.}, vol. 13, no. 2, pp. 902-912, Feb. 2014.

   \bibitem{Lee}
      S. H. Lee, R. Zhang, and K. B. Huang, ``Opportunistic wireless energy harvesting in cognitive radio networks,'' \emph{IEEE Trans. Wireless Commun.}, vol. 12, no. 9, pp. 4788-4799, Sep. 2013.

   \bibitem{Ju_Throughput_for_WPCN}
      H. Ju and R. Zhang, ``Throughput maximization in wireless powered communication networks,'' \emph{IEEE Trans. Wireless Commun.,} vol. 13, no. 1, pp. 418-428, Jan. 2014.



   \bibitem{Sendonaris}
      A. Sendonaris, E. Erkip, and B. Aazhang, ``User cooperation diversity: Part I and Part II,'' \emph{IEEE Trans. Commun.}, vol. 51, no. 11, pp. 1927-1948, Nov. 2003.


   \bibitem{Liang}
      Y. Liang and V. V. Veeravalli, ``Gaussian orthogonal relay channels: optimal resource allocation and capacity,'' \emph{IEEE Trans. Inform. Theory}, vol. 51, no. 9, pp. 3284-3289, Sep. 2005.

   \bibitem{C_Huang}
      C. Huang, R. Zhang, and S. Cui, ``Throughput maximization for the Gaussian relay channel with energy harvesting constraints,'' \emph{IEEE J. Sel. Areas Commun.}, vol. 31, no. 8, pp. 1469-1479, Aug. 2013.

   \bibitem{Gurakan}
      B. Gurakan, O. Ozel, J. Yang, and S. Ulukus, ``Energy cooperation in energy harvesting wireless communications,'' in \emph{Proc. IEEE Inter. Symp. Inform. Theory}, Cambridge, MA, USA, July 2012.

   \bibitem{Nasir}
      A. A. Nasir, X. Zhou, S. Durrani, and R. A. Kennedy, ``Relaying protocols for wireless energy harvesting and information processing,'' \emph{IEEE Trans. Wireless Commun.}, vol. 12, no. 7, pp. 3622-3636, July 2013.



   \bibitem{Ju_ArXiv}
      H. Ju and R. Zhang, ``User cooperation in wireless powered communication networks,'' \emph{ArXiv preprint}, available online at arXiv:1403:7123.


   \bibitem{ConvexOptimization}
      S. Boyd and L. Vandenberghe, \emph{Convex Optimization}. Cambridge University Press, 2004.

   \bibitem{LectureNote}
      S. Boyd, EE364b Lecture Notes. Stanford, CA: Stanford Univ., avaliable online at ${\rm{http://www}}{\rm{.stanford.edu/class/ee364b/lectures/}}$ ${\rm{ellipsoid\_method\_slides}}{\rm{.pdf}}$.




\end{thebibliography}
\end{document}